\documentclass{llncs}
\newcommand{\PAPER}[1]{}
\newcommand{\LNCS}[1]{#1}

\PAPER{
\usepackage{fullpage,amsthm,amsmath,amssymb,mathtools,float,subcaption,hyperref}
}
\LNCS{
\usepackage{mathtools,float,subcaption,hyperref}
}

\PAPER{
\theoremstyle{plain}
\newtheorem{theorem}{Theorem}[section]
\newtheorem{lemma}[theorem]{Lemma}
\newtheorem{claim}{Claim}

\newenvironment{remark}{\paragraph{Remark.}}{}
\newenvironment{remarks}{\paragraph{Remarks.}}{}
}
\LNCS{
\renewenvironment{remark}{\paragraph{Remark.}}{}
\newenvironment{remarks}{\paragraph{Remarks.}}{}
}

\def\Seq#1{\langle #1 \rangle}
\def\Set#1{\left\{ #1 \right\}}

\def\ceil#1{\lceil #1 \rceil}
\def\Floor#1{\left\lfloor #1 \right\rfloor}
\def\Ceil#1{\left\lceil #1 \right\rceil}

\PAPER{
\title{Improved Upper and Lower Bounds for LR Drawings of\\ Binary Trees}

\author{Timothy M. Chan$^*$
\and
Zhengcheng Huang\thanks{Department of Computer Science, University of Illinois at Urbana-Champaign, \{tmc,zh3\}@illinois.edu}}
}

\LNCS{
\title{Improved Upper and Lower Bounds for\\ LR Drawings of Binary Trees}

\author{Timothy M. Chan\inst{1}
\and
Zhengcheng Huang\inst{1}
}

\institute{Department of Computer Science, University of Illinois at Urbana-Champaign
\email{\{tmc,zh3\}@illinois.edu}}

}

\begin{document}

\maketitle

\begin{abstract}
  In SODA'99, Chan introduced a simple type of planar straight-line upward order-preserving drawings of binary trees, known as \emph{LR drawings}: such a drawing is obtained by picking a root-to-leaf path, drawing the path as a straight line, and recursively drawing the subtrees along the paths.  Chan proved that any binary tree with $n$ nodes admits an LR drawing with $O(n^{0.48})$ width.  In SODA'17, Frati, Patrignani, and Roselli proved that there exist families of $n$-node binary trees for which any LR drawing has $\Omega(n^{0.418})$ width.  In this paper, we improve Chan's upper bound to $O(n^{0.437})$ and Frati {\it et al.}'s lower bound to $\Omega(n^{0.429})$. 
    
\end{abstract}

\section{Introduction}


Drawings of trees on a grid with small area have been extensively studied in the graph drawing literature
\cite{Bac,BieJGAA17,BieCCCG17,SODA99,GD96,Cre,CrescenziP97,CrePen,Fra07,%
FratiPR20,GarSoCG93,GarRus03,GarRusICCSA,%
GarRus04,LeeTHESIS,Lei,ReiTil,ShKiCh,ShiIPL,ShiTHESIS,Tre,Val} (see also the book~\cite{DiBOOK} and a recent survey~\cite{DiFraSURV}). 

In this paper, we focus on one simple type of drawings of binary trees called \emph{LR drawings}, which was introduced by Chan in SODA'99~\cite{SODA99} (and named in a later paper by Frati, Patrignani, and Roselli~\cite{FratiPR20}):  For a given binary tree $T$, we place the root somewhere on the top side of the bounding box, recursively draw its left subtree $L$ and its right subtree $R$, and combine the two drawings by applying one of two rules.  In the \emph{left rule}, we connect the root of $T$ to the root of $R$ by a vertical line segment, 
place the bounding box of $L$'s drawing one unit to the left of the vertical line segment, and place the bounding box of $R$'s drawing underneath.  In the \emph{right rule}, we connect the root of $T$ to the root of $L$ by a vertical line segment, place the bounding box of $R$'s drawing one unit to the right of the vertical line segment, and place the bounding box of $L$'s drawing underneath.  See Figure~\ref{fig:lr}(a).  LR drawings are precisely those that can be obtained by recursive applications of these two rules.

\PAPER{
\begin{figure}
  \begin{subfigure}{0.55\textwidth}
  \includegraphics[scale=0.7]{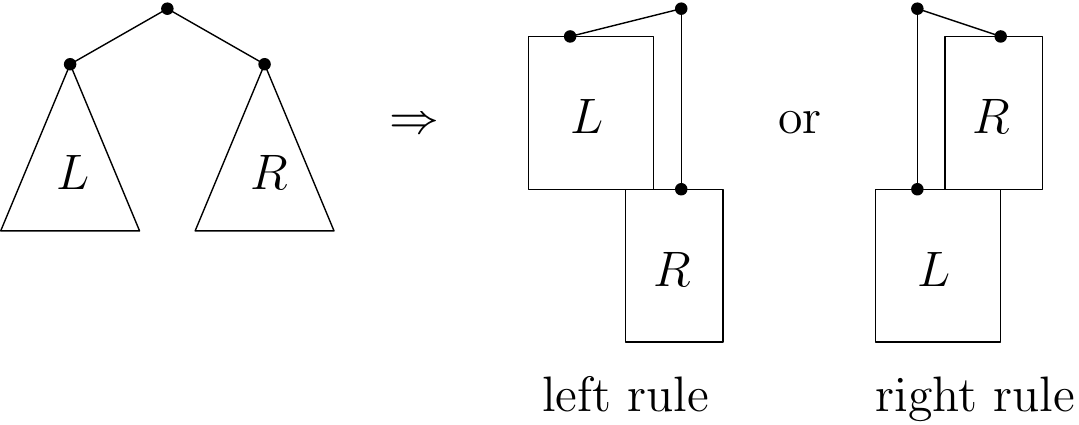}
 \end{subfigure}\hfill
 \begin{subfigure}{0.4\textwidth}
   \includegraphics[scale=0.8]{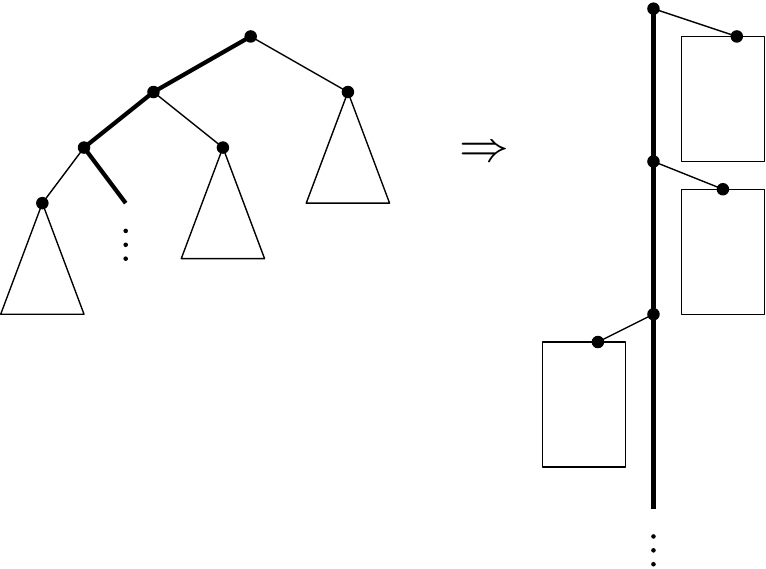}
   \end{subfigure}
  \caption{LR drawing}
  \label{fig:lr}
\end{figure}
}
\LNCS{
\begin{figure}
  \begin{subfigure}[b]{0.55\textwidth}
  \includegraphics[scale=0.52]{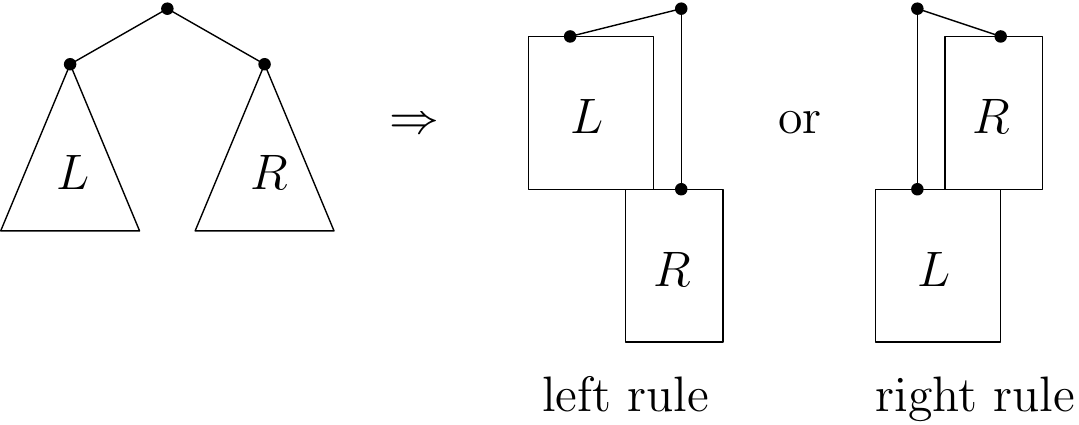}
  \vspace{5ex}
  \caption{}
 \end{subfigure}\hfill
 \begin{subfigure}[b]{0.42\textwidth}
   \includegraphics[scale=0.65]{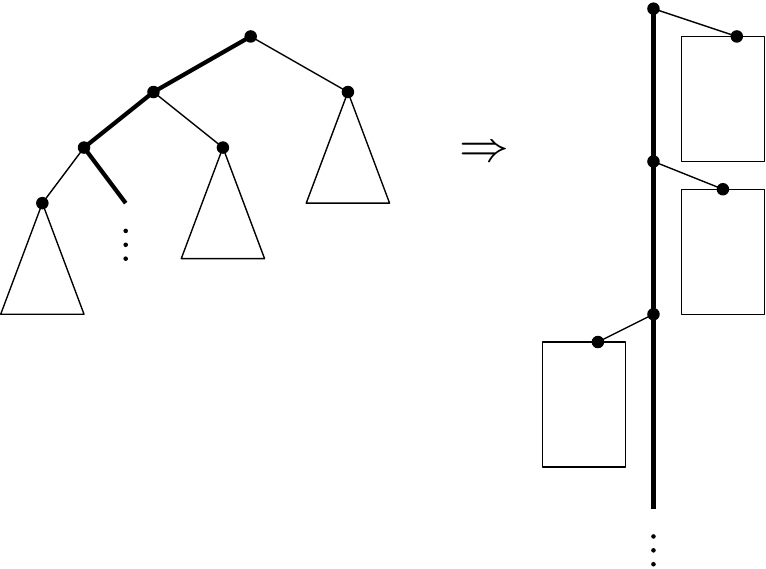}
     \caption{}
   \end{subfigure}
  \caption{LR drawing}
  \label{fig:lr}
\end{figure}

}


(For historical context, we should mention that a similar notion of \emph{hv drawings} were proposed before in some of the early papers on tree drawings~\cite{Cre,CrescenziP97,CrePen}, and were also defined recursively using two rules; the key differences are that in hv drawings, the root is always placed at the upper left corner, and the order of the left and right subtrees may not be preserved.)

Alternatively, LR drawings have the following equivalent definition: for a given binary tree $T$, we pick a root-to-leaf path $\pi$, draw $\pi$ on a vertical line, and recursively draw all subtrees of $\pi$ (i.e., subtrees rooted at siblings of the nodes along $\pi$), placing the bounding boxes of left subtrees of $\pi$ one unit to the left of the vertical line, and the bounding boxes of the right subtrees of $\pi$ one unit to the right of the vertical line.  See Figure~\ref{fig:lr}(b).  


It is easy to see that LR drawings satisfy the following desirable properties:
\begin{enumerate}
  \item {\it Planar}: edges do not cross in the drawing.
  \item {\it Straight-line}: edges are drawn as straight line segments.
  \item {\it Strictly upward}: a parent has a strictly larger $y$-coordinate than each child.
  \item {\it Order-preserving}: the edge from a parent to its left child is to the left of the edge from the parent to its right child in the drawing.
 \end{enumerate}
Indeed, the original motivation for LR drawings is in finding ``good'' planar, straight-line, strictly upward, order-preserving drawings of binary trees~\cite{SODA99}.  Goodness here is measured in terms of the \emph{area} of a drawing, defined as the width (the number of grid columns) times the height (the number of grid rows),
assuming that nodes are placed on an integer grid.  The goal is to prove worst-case bounds on the minimum area needed for such drawings as a function of the number of nodes $n$.
As $\Omega(n)$ height is clearly necessary in the worst case for strictly upward drawings (and LR drawings have $O(n)$ height), the goal becomes bounding the width. Chan's original paper gave several methods to produce LR drawings of arbitrary binary trees, the first method guaranteeing $O(n^{0.695})$ width, a second method with $O(\sqrt{n})$ width, and a final method (described in the appendix of his paper) with $O(n^{0.48})$ width.

More recently, in   SODA'17, Frati, Patrignani, and Roselli~\cite{FratiPR20} proved the first nontrivial lower bound, showing that there exist binary trees for which any LR drawing requires
$\Omega(n^{0.418})$ width.  This raises an intriguing question: can the gap between upper and lower bounds be closed, and the precise value of the exponent be determined?  

It should be mentioned that other methods were subsequently found for planar, straight-line, strictly upward, order-preserving drawings of binary trees with smaller width ($2^{O(\sqrt{\log n})}$ in Chan's original paper, and 
eventually, $O(\log n)$  in a paper by Garg and Rusu~\cite{GarRus03}).  Nevertheless, the question on LR drawings is still interesting and natural, as it is fundamentally about combinatorics of trees, or more specifically, decompositions of trees via path separators (instead of the more usual vertex or edge separators).  Indeed, by the alternative definition,
the minimum LR-drawing width $W^*(T)$ of a binary tree~$T$ can be described by the following  self-contained formula, without reference to geometry:

\[ W^*(T) \:=\: \min_\pi \max_{\alpha,\beta}\: (W^*(\alpha)+W^*(\beta)+1),
\]
where the minimum is over all root-to-leaf paths $\pi$ in $T$, and the maximum is over all left subtrees $\alpha$ of $\pi$ and all right subtrees $\beta$ of $\pi$. 

The LR drawing problem was also mentioned 
in Di Battista and Frati's recent survey~\cite{DiFraSURV} (as ``Open problem 10'').\footnote{Technically, that survey asks about a different but related function: $W^{**}(n) = \max_T \min_\pi\max_{\alpha,\beta} (W^{**}(|\alpha|) + W^{**}(|\beta|) + 1)$, where the outer maximum is over all $n$-node binary trees~$T$.  This function may be larger than $\max_{T: |T|=n} W^*(T)$.}
LR drawing techniques have been applied to solve other problems, for example, on \emph{octagonal},\footnote{All edges have slope from $\{0,\pm 1,\pm\infty\}$.} planar, straight-line, strictly upward, order-preserving drawings of binary trees~\cite{SODA99},
\emph{orthogonal},\footnote{All edges are horizontal or vertical.} planar, straight-line, \emph{non-upward}, order-preserving drawings of binary trees~\cite{Fra07},
and planar straight-line drawings of \emph{outerplanar graphs}~\cite{GargR07,BattistaF09}, although in each of these applications, better methods not relying on LR drawings were eventually found \cite{BieCCCG17,Chan18,FratiPR20}.

In this paper, we make progress in narrowing the gap on the width bounds for LR drawings of binary trees:
we improve Chan's upper bound from $O(n^{0.48})$ to $O(n^{0.437})$, and improve Frati et al.'s lower bound from $\Omega(n^{0.418})$ to $\Omega(n^{0.429})$.

\section{Upper Bound}

In this section, we present an algorithm for LR drawings that achieves width $O(n^{0.438})$. A small improvement to $O(n^{0.437})$ will be given in the next section.
Our algorithm builds upon Chan's approach \cite[Appendix~A]{SODA99} but uses new ideas to substantially improve his $O(n^{0.48})$ upper bound. 
Throughout the paper,
let $|T|$ denote the size (i.e., the number of nodes) in a tree~$T$.

\subsection{The Algorithm}

Given a binary tree $T$ with $n$ nodes, we describe a recursive algorithm to produce an LR drawing of $T$ and show by induction that its width is at most $cn^p$, for some constants $p$ and $c$ to be set later.

For $n$ smaller than a sufficiently large constant, we can draw $T$ arbitrarily.  Otherwise,
we maintain a path $\pi=\Seq{v_0,\ldots,v_t}$. 
A {\em subtree of $\pi$} refers to a subtree rooted at a sibling of a node in $\pi$  (it does not include the two subtrees at $v_t$).
Let $\alpha$ and $\beta$ denote the largest left subtree and right subtree of~$\pi$, respectively.
We maintain the invariant that 

$$|\alpha|^p+|\beta|^p\le (1-\delta)n^p$$
for some sufficiently small constant $\delta>0$.  Initially, $t=0$ and $v_0$ is the root of $T$. If $v_t$ is a leaf, then we draw the subtrees of $\pi$ recursively and combine them by aligning $\pi$ vertically; the width is bounded by $c|\alpha|^p+c|\beta|^p+1$, which by the invariant (and the induction hypothesis) is at most $c(1-\delta)n^p+1 < cn^p$ for a sufficiently large $c$ (depending on $p$ and $\delta$). 
From now on, assume that $v_t$ is not a leaf.  Let $L$ and $R$ be the left and right subtree of the current node $v_t$, respectively.  For some choice of constants $\delta>0$ and $h$, we consider four cases (which cover all possibilities, as we will show in the next subsection). 

\medskip
\noindent{\sc Case 1:}
$|\alpha|^p+|R|^p\leq(1-\delta)n^p$.
  \label{cas:1}
  Set $v_{t+1}$ to be the left child of $v_t$. Increment $t$ and repeat.

\medskip
\noindent{\sc Case 2:}
$| \beta|^p+|L|^p\leq(1-\delta)n^p$.
  \label{cas:2}
  Set $v_{t+1}$ to be the right child of $v_t$. Increment $t$ and repeat.

\medskip
In either of the above two cases, the invariant is clearly preserved.

We may now assume that
$|\alpha|^p+|R|^p>(1-\delta)n^p$ and
$|\beta|^p+|L|^p>(1-\delta)n^p$.
In conjunction with the invariant, we know that
$|\beta|<|R|$ and $|\alpha|<|L|$.

For the next two cases, 
we introduce notation for the left and right subtrees of $\pi$ (see Figure~\ref{fig:left-sub-notation}).
Let $\alpha^{(0)}_1=\alpha$ (the largest left subtree of $\pi$). 
The parent of $\alpha^{(0)}_1$ divides $\pi$ into two segments.
Let $\alpha^{(1)}_1$ and $\alpha^{(1)}_2$ denote the largest left subtree of the top and bottom segment, respectively.
Extend the definition analogously:
For each $i$, the parents of the $2^i-1$ subtrees in $\{ \alpha^{(\ell)}_j \mid 0\leq\ell<i,\ 1\leq j\leq 2^\ell\}$ divide $\pi$ into $2^i$ segments. In the downward order, let $\alpha^{(i)}_1,\ldots,\alpha^{(i)}_{2^i}$ denote the largest left subtrees of these segments.  The above labeling of subtrees resembles a ``ruler pattern'' (like in~\cite{FratiPR20}).
We define the right subtrees $\beta^{(i)}_1,\ldots,\beta^{(i)}_{2^i}$  similarly (we do not care how the left subtrees and the right subtrees of $\pi$ interleave).

\begin{figure}[H]
  \centering
  \includegraphics[width=.17\linewidth]{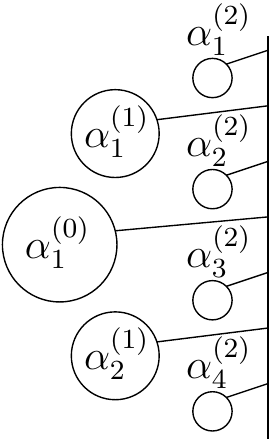}
  \caption{Notation for left subtrees}
  \label{fig:left-sub-notation}
\end{figure}

\medskip
\noindent{\sc Case 3:}
There exists $i\leq h$ such that $\sum_{j=1}^{2^i} |\alpha^{(i)}_j|^p+\max\{|L|,|R|\}^p\leq(1-\delta)n^p$.
  \label{cas:3}
  We generate an LR drawing of $T$ using a procedure called the $i${\it-right-twist}:  We bend $\pi$ at the parents of all subtrees in $\{ \alpha_j^{(\ell)}\mid 0\le \ell< i,\ 1\le j\le 2^\ell\}$
  (all these subtrees are thus pulled downward in the drawing), as illustrated in Figure \ref{fig:right-twist}.
  We recursively draw $R$.
  We draw most of the subtrees of $\pi$ recursively as well, but with the following exceptions: for the subtrees in $\{ \alpha_j^{(\ell)}\mid 0\le \ell< i,\ 1\le j\le 2^\ell\}$, we make their leftmost paths vertically aligned and recursively draw the subtrees of these paths.  
  Similarly, for $L$, we make its leftmost path vertically aligned and recursively draw the subtrees of the path.
  Since every subtree of $\pi$ has size at most $\max\{|\alpha|,|\beta|\}<\max\{|L|,|R|\}$, it is easy to check (using the induction hypothesis) that the
  resulting LR drawing has width at most
  $\sum_{j=1}^{2^i}c|\alpha^{(i)}_j|^p+c\max\{|L|,|R|\}^p + 2^{h}$; this is at most $(1-\delta)cn^p + 2^{h} < cn^p$, for a sufficiently large~$c$ (depending on $p$, $\delta$, and $h$).

\medskip
\noindent{\sc Case 4:}
There exists $i\leq h$ such that $\sum_{j=1}^{2^i} |\beta^{(i)}_j|^p+\max\{|L|,|R|\}^p\leq(1-\delta)n^p$.
  \label{cas:4}
  This is similar to Case 3, by using the
  $i${\it-left-twist}.

\begin{figure}
  \centering
  \begin{subfigure}{.26\textwidth}
    \centering
    \includegraphics[width=\textwidth]{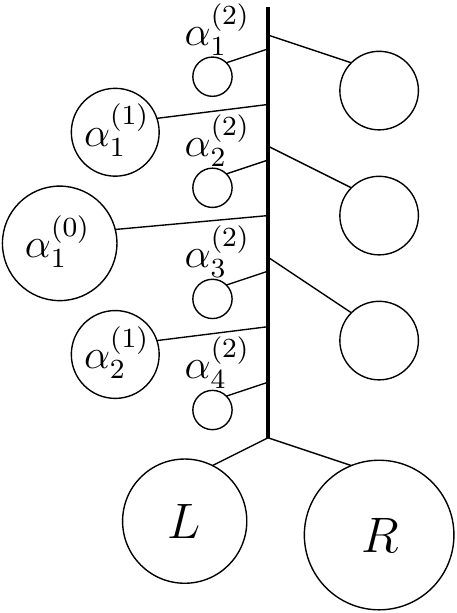}
    \caption{Before right-twist}
  \end{subfigure}
  \hspace{.1\textwidth}
  \begin{subfigure}{.22\textwidth}
    \centering
    \includegraphics[width=\textwidth]{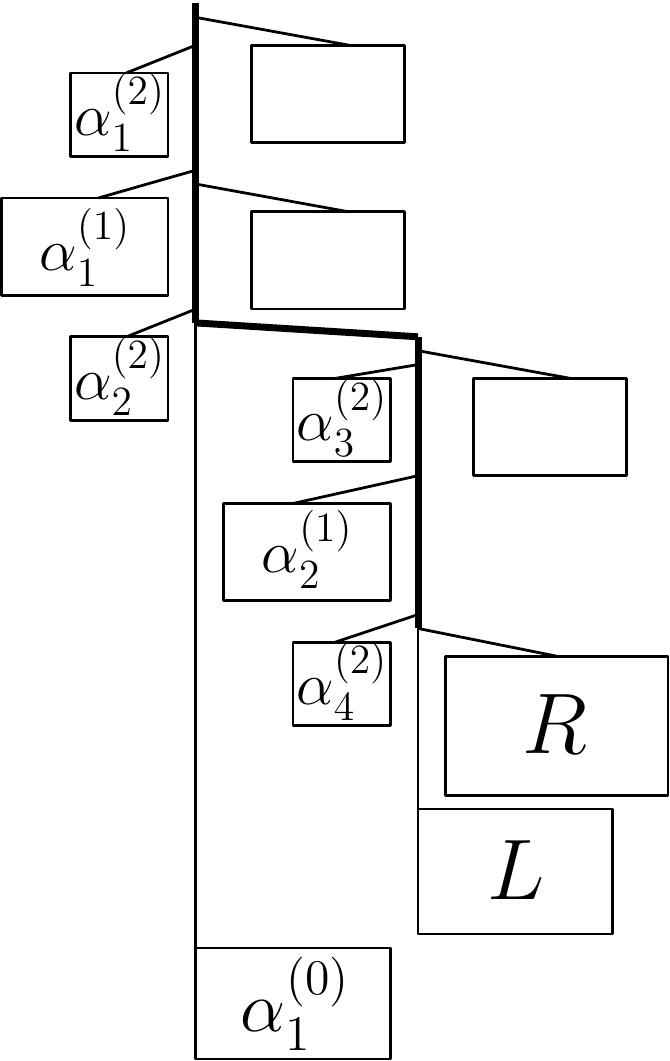}
    \caption{After $1$-right-twist}
  \end{subfigure}
  \hspace{.1\textwidth}
  \begin{subfigure}{.26\textwidth}
    \centering
    \includegraphics[width=\textwidth]{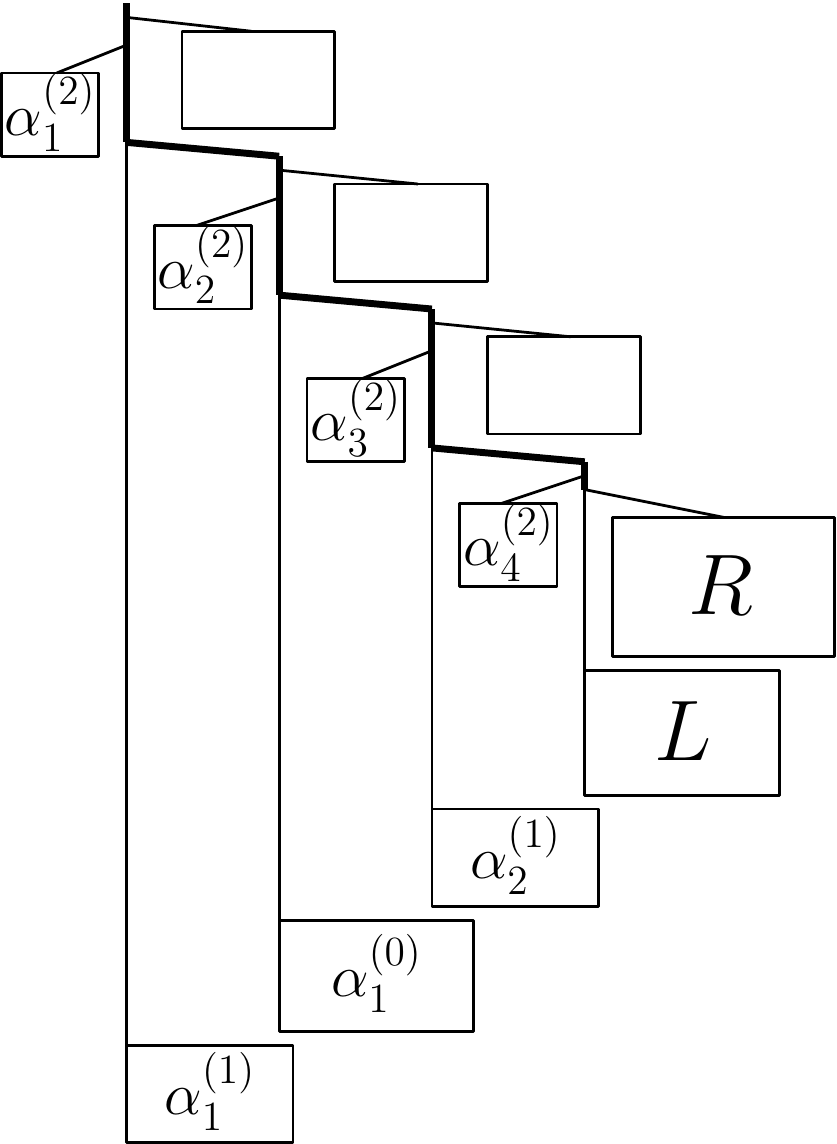}
    \caption{After $2$-right-twist}
  \end{subfigure}
  \caption{Right twist}
  \label{fig:right-twist}
\end{figure}

\begin{remark}
The twisting procedures in Cases 3 and~4, and the introduction of the ``ruler pattern'', are the main new ideas, compared to Chan's previous algorithm~\cite{SODA99}.
\end{remark}

\subsection{Analysis}

To complete the induction proof, it suffices to show that these four cases cover all possibilities.

\begin{lemma}
  \label{lem:complete}
  For $p=0.438$ and a sufficiently small constant $\delta>0$ and a sufficiently large constant~$h$,
  
  \begin{equation*}
    \label{eqn:complete}
    \min\Set{\begin{array}{l}
      |\alpha|^p+|R|^p, \\
      | \beta|^p+|L|^p, \\
      \displaystyle \min_{i=1}^h \left(\sum_{j=1}^{2^i} |\alpha^{(i)}_j|^p + \max\{|L|,|R|\}^p\right), \\
      \displaystyle \min_{i=1}^h \left( \sum_{j=1}^{2^i} | \beta^{(i)}_j|^p + \max\{|L|,|R|\}^p\right)
    \end{array}} \:\leq\: (1-\delta)n^p.
  \end{equation*}
\end{lemma}

\begin{proof}
  Assume for the sake of contradiction that the lemma 
  is false. Without loss of generality, assume $|R|\ge |L|$.  Let $a_0,\ldots,a_h,b_0,\ldots,b_h$ be positive real numbers with $\sum_{i=0}^h a_i + \sum_{i=0}^h b_i=1$, whose values are to be determined later.  Let
  
\PAPER{  
\[
    X \: := \: a_0 ( |\alpha|^p + |R|^p) +
    b_0 ( | \beta|^p + |L|^p ) + 
    \sum_{i=1}^h a_i \left(
        \sum_{j=1}^{2^i}|\alpha^{(i)}_j|^p+|R|^p
      \right)
      +
      \sum_{i=1}^h b_i \left(
        \sum_{j=1}^{2^i}|\beta^{(i)}_j|^p+|R|^p
      \right).
     \]
}
\LNCS{
\begin{eqnarray*}
    X &:=& a_0 ( |\alpha|^p + |R|^p) +
    b_0 ( | \beta|^p + |L|^p ) + {}\\&&
    \sum_{i=1}^h a_i \left(
        \sum_{j=1}^{2^i}|\alpha^{(i)}_j|^p+|R|^p
      \right)
      +
      \sum_{i=1}^h b_i \left(
        \sum_{j=1}^{2^i}|\beta^{(i)}_j|^p+|R|^p
      \right).
\end{eqnarray*}
}
     By our assumption, $X > (1-\delta)n^p$.  On the other hand,
  by H\"older's inequality,%
  \footnote{H\"older's inequality states that  $\sum_i |x_iy_i|\le \left(\sum_i |x_i|^s\right)^{1/s}  \left(\sum_i |y_i|^t\right)^{1/t}$ for any $s,t>1$ with $\frac 1s+\frac 1t=1$.  In our applications, it is more convenient to set $s=\frac 1p$, $t=\frac 1{1-p}$,
  $x_i=X_i^p$, and $y_i=c_i$, and rephrase the inequality as:
  $\sum_i c_iX_i^p \le \left(\sum_i c_i^{1/(1-p)}\right)^{1-p} \left(\sum_i X_i\right)^p$ for any $0<p<1$ and $c_i,X_i\ge 0$.
  }
  
  \begin{eqnarray*}
      X &=& 
      a_0 |\alpha|^p + b_0 |\beta|^p + b_0 |L|^p +  \left(1-b_0\right)|R|^p +
      \LNCS{{}\\&&}
      \sum_{i=1}^h a_i       \sum_{j=1}^{2^i}|\alpha^{(i)}_j|^p+\sum_{i=1}^h b_i       \sum_{j=1}^{2^i}|\beta^{(i)}_j|^p
      \\
        &\leq& \lambda^{1-p} \left( |\alpha|+|\beta|+|L|+|R| +
          \sum_{j=1}^h \sum_{j=1}^{2^i}|\alpha^{(i)}_j| +
          \sum_{j=1}^h \sum_{j=1}^{2^i}|\beta^{(i)}_j| 
        \right)^p \\
      &\leq& \lambda^{1-p}n^p,
  \end{eqnarray*}
  where
  
  \begin{equation*}
    \label{eqn:lambda-1}
    \lambda\: :=\:
      a_0^\frac{1}{1-p} + 2b_0^{\frac{1}{1-p}} +  
      \left( 1- b_0
      \right)^\frac{1}{1-p} +
       \sum_{i=1}^h 2^i a_i^\frac{1}{1-p} +
       \sum_{i=1}^h 2^i b_i^\frac{1}{1-p}.
  \end{equation*}
  Thus, we have $\lambda^{1-p} > 1-\delta$. 
  However, we show that this is not true for some choice of parameters. 

  We first set $a_i=b_i=(2^{-\frac{1-p}{p}})^i a_0$ for $1\leq i\leq h$ (by calculus, this choice is actually the best for minimizing $\lambda$). 
  Let $\rho=1+2\sum_{i=1}^{h}(2^{-\frac{1-p}{p}})^i$ and
  $b_0=1-\rho a_0$.  Then we indeed have $\sum_{i=0}^h a_i + \sum_{i=0}^h b_i=1$, and the above expression simplifies to
  $\lambda=\rho a_0^\frac{1}{1-p}+2(1-\rho a_0)^\frac{1}{1-p}+ (\rho a_0)^\frac{1}{1-p}$.
  For $p=0.438$, the limit of $\rho$ (as $h\rightarrow\infty$) is $2/(1-2^{-\frac{1-p}{p}})-1\approx 2.395068$.  
  We can plug in $a_0=0.247$ and
  verify (using a calculator) that the limit of
  $\lambda$ is less than $0.9984$, which leads to a contradiction for
  a sufficiently small
  $\delta$ and a sufficiently large $h$.
\end{proof}

\section{Slightly Improved Upper Bound}\label{sec:slight}

In this section, we describe a refinement of our algorithm to further improve the width upper bound to $O(n^{0.437})$.  Although the improvement is tiny, the main purpose is to show that our algorithm is not optimal.

The change lies in the procedure of $i$-right-twist in Case~3, specifically, how $L$ is drawn.  Instead of vertically aligning the leftmost path in $L$, we choose a different path, exploiting the already ``used'' width from the drawing of $\alpha^{(i)}_{2^i}$ that is available to the left of the root of $L$.  We define a new path $\pi'=\Seq{u_0,u_1,\ldots}$ in $L$ as follows.  Initially, set $u_0$ to the root of $L$. For $k=0,1,\ldots$ (until $u_k$ is a leaf), if the left subtree of $u_k$ has size at most
$|\alpha^{(i)}_{2^i}|$, then set $u_{k+1}$ to be the right child of $u_k$;
otherwise, set $u_{k+1}$ to be the left child of $u_k$ (see Figure \ref{fig:left-twist-adv}).
This way, every left subtree of $\pi'$ has size at most $|\alpha^{(i)}_{2^i}|$, and every right subtree of $\pi'$ has size less than $|L|-|\alpha^{(i)}_{2^i}|$.
(Note that $|\alpha^{(i)}_{2^i}| < |L|$.)
We draw $L$ by vertically aligning the path $\pi'$, and recursively drawing the left and right subtrees of $\pi'$.
The overall LR drawing of $T$ has width at most
$\sum_{j=1}^{2^i}c|\alpha^{(i)}_j|^p+c\max\{|L|-|\alpha^{(i)}_{2^i}|,|R|\}^p + 2^{h}$.
(Parts of the drawing may have width bounded instead by 
$\sum_{j=1}^{2^i-1}c|\alpha^{(i)}_j|^p + c\max\{|L|,|R|\}^p + 2^h$,
but this is no worse than the above bound since $|L|^p\le |\alpha^{(i)}_{2^i}|^p + (|L|-|\alpha^{(i)}_{2^i}|)^p$.)
Thus, we can relax the condition in Case 3 to
$\sum_{j=1}^{2^i}c|\alpha^{(i)}_j|^p+c\max\{|L|-|\alpha^{(i)}_{2^i}|,|R|\}^p \le (1-\delta)n^p$.

\begin{figure}
  \centering
  \includegraphics[width=.24\linewidth]{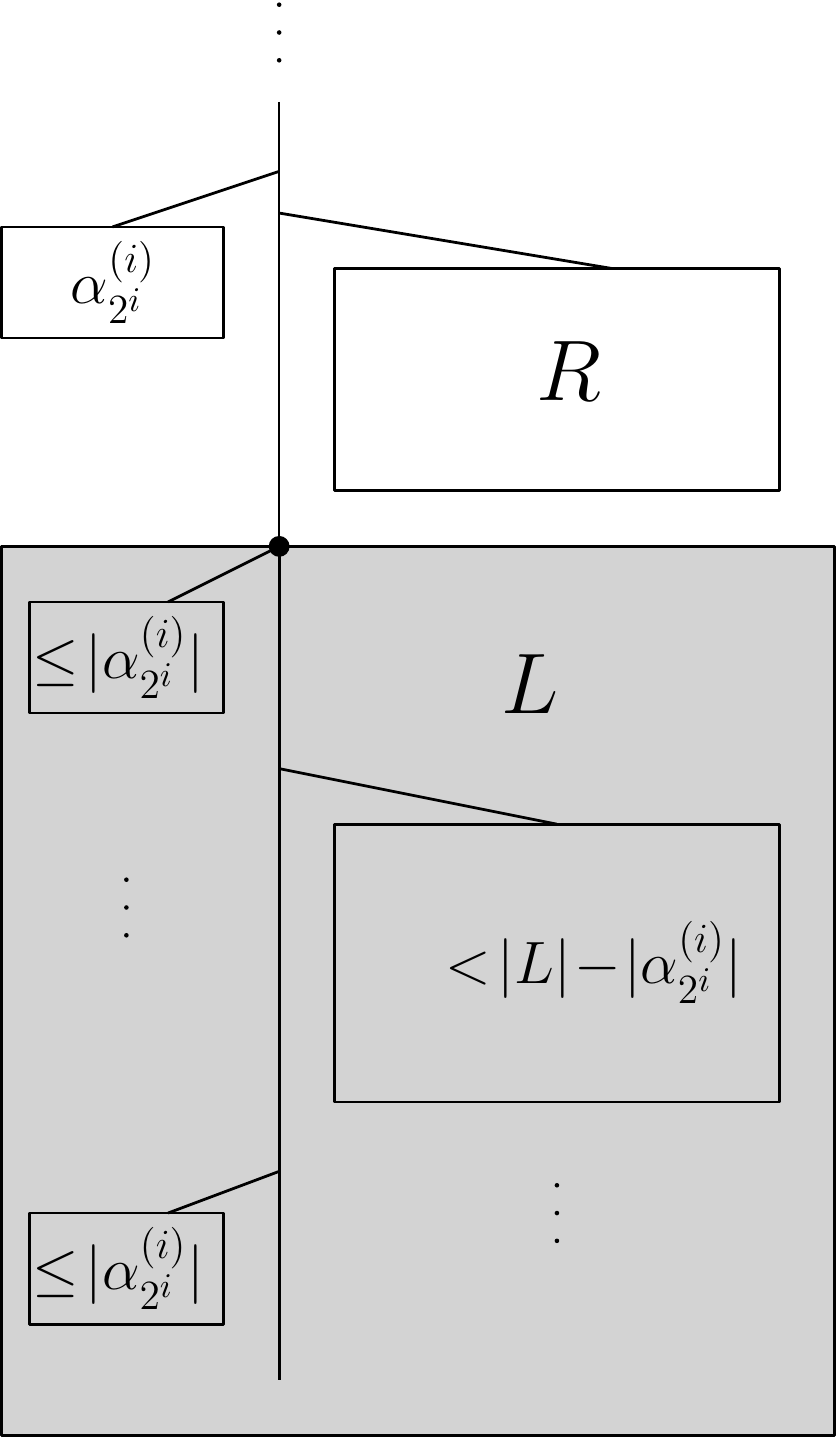}
  \caption{Choosing a path $\pi'$ inside $L$}
  \label{fig:left-twist-adv}
\end{figure}

With a similar modification to the $i$-left-twist procedure,
we can relax the condition in Case~4 to
$\sum_{j=1}^{2^i}c|\beta^{(i)}_j|^p+c\max\{|L|,|R|-|\beta^{(i)}_{2^i}|\}^p \le (1-\delta)n^p$.

It suffices to prove the following variant of Lemma~\ref{lem:complete}:

\begin{lemma}
  \label{lem:complete:refined}
  For $p=0.437$ and a sufficiently small constant $\delta>0$ and a sufficiently large constant~$h$,

  \begin{equation*}
    \label{eqn:complete:refined}
    \min\Set{\begin{array}{l}
      |\alpha|^p+|R|^p, \\
      | \beta|^p+|L|^p, \\
      \displaystyle \min_{i=1}^h \left(      \sum_{j=1}^{2^i} |\alpha^{(i)}_j|^p + \max\{|L|-|\alpha^{(i)}_{2^i}|,|R|\}^p\right), \\
      \displaystyle \min_{i=1}^h \left( \sum_{j=1}^{2^i} | \beta^{(i)}_j|^p + \max\{|L|,|R|-|\beta^{(i)}_{2^i}|\}^p \right)
    \end{array}} \:\leq\: (1-\delta)n^p.
  \end{equation*}
\end{lemma}
\begin{proof}
  Assume for the sake of contradiction that the lemma 
  is false. Without loss of generality, assume $|R|\ge |L|$.  
  Note that $|\beta^{(i)}_{2^i}|$ decreases with $i$. Let $i^*$ be the smallest integer with $1\le i^*\le h$ such that $|R|-|\beta^{(i^*)}_{2^{i^*}}|\geq|L|$ (if such an integer does not exist, set $i^*=h+1$).
  
  Let $a_0,\ldots,a_h,b_0,\ldots,b_h$ be positive real numbers with $\sum_{i=0}^h a_i+ \sum_{i=0}^h b_i=1$, whose values are to be determined later.  Let
  
  \begin{equation*}
  \begin{aligned}
    X\: := \: \sum_{i=0}^h a_i \left( \sum_{j=1}^{2^i}|\alpha^{(i)}_j|^p+|R|^p \right)
    &+ \sum_{i=0}^{i^*-1} b_i \left( \sum_{j=1}^{2^i}|\beta^{(i)}_j|^p+|L|^p \right) \\
    &+ \sum_{i=i^*}^h b_i \left( \sum_{j=1}^{2^i}|\beta^{(i)}_j|^p+(|R|-|\beta^{(i)}_{2^i}|)^p \right).   
  \end{aligned}
\end{equation*}
     By our assumption, $X > (1-\delta)n^p$.  On the other hand,
  by H\"older's inequality, for any $0<\gamma<1$,
  
\begin{equation*}
  \begin{aligned}
    |\beta^{(i)}_{2^i}|^p+(|R|-|\beta^{(i)}_{2^i}|)^p
    &= (1-\gamma)|\beta^{(i)}_{2^i}|^p+\gamma|\beta^{(i)}_{2^i}|^p+(|R|-|\beta^{(i)}_{2^i}|)^p \\
    &\le (1-\gamma)|\beta^{(i)}_{2^i}|^p+(\gamma^\frac{1}{1-p}+1)^{1-p}|R|^p,
  \end{aligned}
\end{equation*}
and by H\"older's inequality again,
  \begin{eqnarray*}
    X &\: \le \: & \sum_{i=0}^h a_i \sum_{j=1}^{2^i}|\alpha^{(i)}_j|^p
    + \sum_{i=0}^{i^*-1} b_i \sum_{j=1}^{2^i}|\beta^{(i)}_j|^p
    + \left( \sum_{i=0}^{i^*-1}b_i \right) |L|^p \LNCS{\\&&}
    + \sum_{i=i^*}^h b_i \left(\sum_{j=1}^{2^i-1}|\beta^{(i)}_j|^p + (1-\gamma) |\beta^{(i)}_{2^i}|^p\right) \\
    &&+ \left( \sum_{i=0}^h a_i +  (1+\gamma^\frac{1}{1-p})^{1-p} \sum_{i=i^*}^h b_i \right) |R|^p\\  &\leq & \lambda^{1-p}n^p,
  \end{eqnarray*}
where
\begin{equation*}
  \begin{aligned}
    \lambda \: := \:
    \sum_{i=0}^h 2^i a_i^\frac{1}{1-p} &+
    \sum_{i=0}^{i^*-1} 2^i b_i^\frac{1}{1-p} +
    \left( \sum_{i=0}^{i^*-1} b_i \right)^\frac{1}{1-p} +
    \sum_{i=i^*}^h (2^i-1+(1-\gamma)^\frac{1}{1-p})b_i^\frac{1}{1-p} \\ & +
    \left( \sum_{i=0}^h a_i +
         (1+\gamma^\frac{1}{1-p})^{1-p}\sum_{i=i^*}^h b_i \right)^\frac{1}{1-p}.
  \end{aligned}
\end{equation*}

An optimal choice of parameters is now messier to describe, but will not be necessary.
We can reuse our earlier choice with
$a_0=0.247$, $a_i=b_i=(2^{-\frac{1-p}{p}})^ia_0$ for $1\le i\le h$, and $b_0=1-\sum_{i=0}^h a_i - \sum_{i=1}^h b_i$.  For $p=0.437$, $\gamma=0.1$, and $h=7$, we can verify (with a short computer program) that for each possible $i^*\in\{1,\ldots,8\}$,  $\lambda$ evaluates to strictly less than~1, a contradiction.
\end{proof}

\begin{theorem}
For any binary tree with $n$ nodes, there exists an LR drawing with $O(n^{0.437})$ width.
\end{theorem}

\begin{remark}
It is not difficult to implement the algorithm to construct the drawing in $O(n)$ time.
\end{remark}

\section{Lower Bound}

We now prove an $\Omega(n^{0.429})$ lower bound on the width of LR drawings. Our proof is largely based on Frati, Patrignani, and Roselli's~\cite{FratiPR20}; we show that a simple variation of their proof is sufficient to improve their $\Omega(n^{0.418})$ lower bound.

\subsection{Tree Construction}

For any given positive integer $n$, we describe a recursive construction of a binary tree $T_n$
with $n$ nodes and show by induction that any LR drawing of $T_n$ has width at least $cn^p$, where $p$ and $c>0$ are constants to be determined later.

\begin{figure}[H]
  \centering
  \includegraphics[width=.36\linewidth]{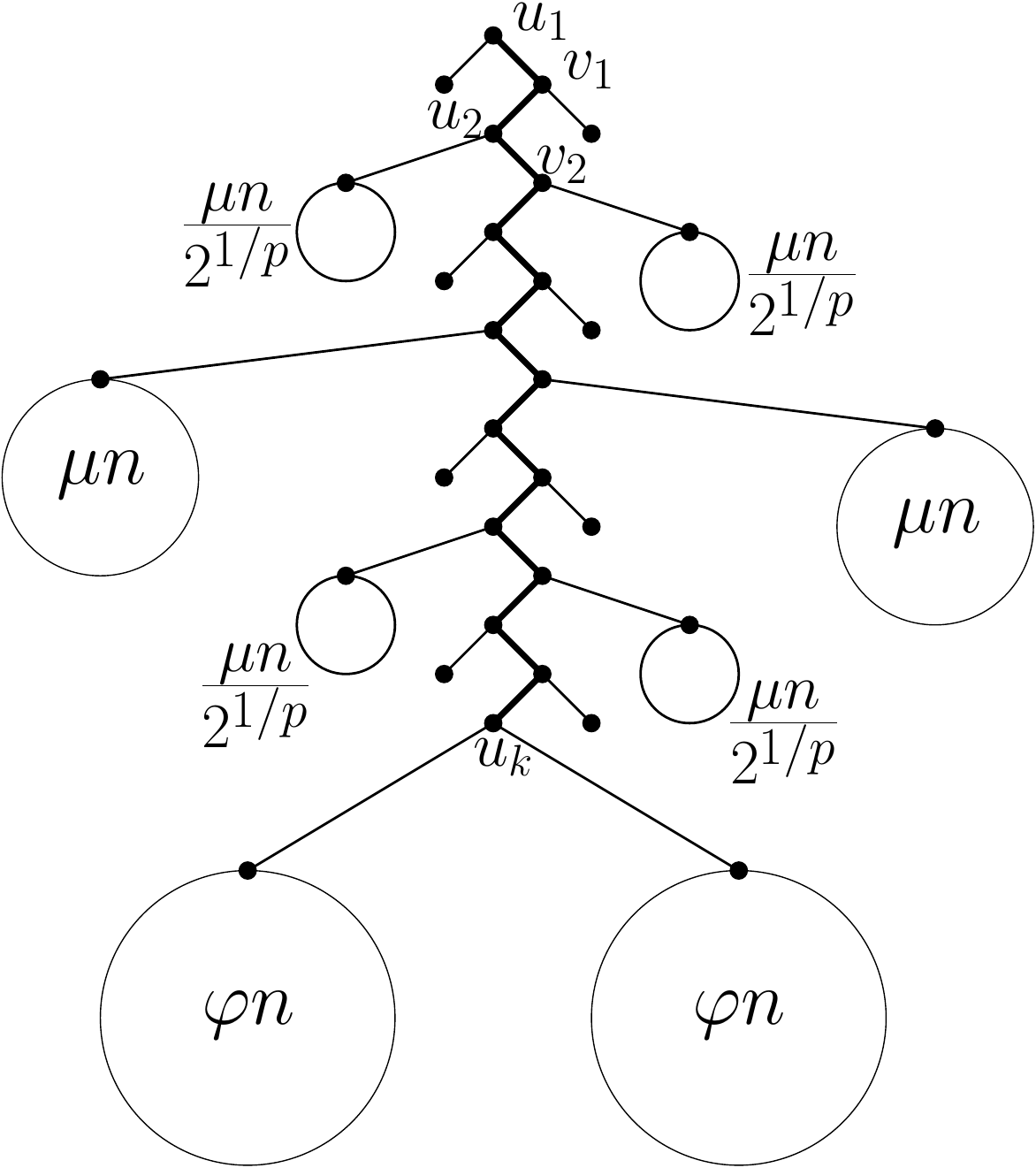}
  \caption{Tree construction for the lower bound}
  \label{fig:lower-bound-tree}
\end{figure}


For $n$ smaller than a sufficiently large constant, we can construct $T_n$ arbitrarily. Otherwise, let
$h$, $\varphi$, and $\mu$ be parameters, to be chosen later.
We construct a tree $T_n$ containing a path $\pi=\Seq{u_1,v_1,u_2,v_2,\ldots,u_{k-1},v_{k-1},u_k}$,
where $k=2^h$ and $u_1$ is the root. The left and right subtree of $u_k$, which we denote by $L$ and $R$, are recursively constructed trees each with $\ceil{\varphi n}$ nodes.  

We will add left subtrees 
$\alpha_1,\ldots,\alpha_{k-1}$ to $u_1,\ldots,u_{k-1}$ and right subtrees $\beta_1,\ldots,\beta_{k-1}$ to $v_1,\ldots,v_{k-1}$.  Specifically, the subtrees $\alpha_{k/2}$ and $\beta_{k/2}$, which are said to be at {\it level $0$}, are recursively constructed trees each with $\Ceil{\mu n}$ nodes.
The subtrees  $\alpha_{k/4},\alpha_{3k/4}$
and $\beta_{k/4},\beta_{3k/4}$, which are at {\it level $1$}, are recursively constructed trees each with $\Ceil{2^{-1/p}\mu n}$ nodes.
Extend the process analogously: For each $i\le h-2$, 
the $2^i$ left subtrees  $\alpha_{k/2^{i+1}},\alpha_{3k/2^{i+1}},\alpha_{5k/2^{i+1}},\ldots$ 
and $2^i$ right subtrees  $\beta_{k/2^{i+1}},\beta_{3k/2^{i+1}},\beta_{5k/2^{i+1}},\ldots$,
which are at {\it level $i$}, are recursively constructed trees each with $\Ceil{(2^{-1/p})^i\mu n}$ nodes.
As shown in Figure \ref{fig:lower-bound-tree}, these subtrees of $\pi$ form a ``ruler pattern'' (which somewhat resembles the ruler pattern from our upper bound proof, coincidentally or not).

We set $h=\Floor{p\log(\mu n/c_0)}$ for a sufficiently large constant $c_0$, and choose parameters $\varphi$ and $\mu$ to satisfy

\begin{equation}
  \label{eqn:linear-constraint}
  \varphi + \frac{\mu}{1-2^{-\frac{1-p}{p}}} \,=\, \frac{1}{2}.
\end{equation}
Then the total size of the left subtrees at level $0,\ldots,h-2$ is
\begin{eqnarray*}
    \sum_{i=0}^{h-2} 2^i\Ceil{(2^{-1/p})^i\mu n}
&=&\sum_{i=0}^{h-2} (2^{-\frac{1-p}{p}})^i\mu n + O(k)
\PAPER{\ =\ }
\LNCS{\\&=&}
\frac{\mu n}{1-2^{-\frac{1-p}{p}}} - \Theta((2^{-\frac{1-p}{p}})^h \mu n) + O(k)\\
&=& \left(\tfrac12 - \varphi\right)n - \Theta(c_0k).
\end{eqnarray*}
The same bound holds for the right subtrees at level $0,\ldots,h-2$.
Thus, we can distribute
$\Theta(c_0)$ nodes to each of the $\Theta(k)$ subtrees at the last level $h-1$ so that $|T_n|$ is exactly~$n$.

\subsection{Analysis}

We begin with a simple property arising from the ruler pattern:

\begin{lemma}
  \label{lem:subtrees}
  For any set $J\subseteq \{1,\ldots,k-1\}$ of consecutive integers, the largest subtree $\alpha_j$ (or $\beta_j$) with $j\in J$ has size at least $ \left(\frac{|J|-1}{k} \right)^{1/p}\mu n$.
\end{lemma}
\begin{proof}
  We may assume $|J|\ge 2$ (for otherwise the inequality is trivial).
  Say $k/2^{i+1}\le |J| < k/2^{i}$.  The subtrees $\alpha_j$ at level at most $i$ are precisely
  those with indices $j$ divisible by $k/2^{i+1}$; 
  there exists one such index with $j\in J$.
  The size of $\alpha_j$ and $\beta_j$ is at least
  $(2^{-1/p})^i \mu n
  \ge (|J|/k)^{1/p} \mu n$.
  \qed
\end{proof}

Assume inductively that any LR drawing of $T_{n'}$ has width at least $c(n')^p$, for all $n'<n$.
Let $T_n(u_j)$ denote the subtree of $T_n$ rooted at node $u_j$.  We will prove the following claim, for $c$ sufficiently small:

\begin{claim}
  \label{claim:wst-indct}
  For $j\in\{1,\ldots,k\}$,  every LR drawing $\Gamma$ of $T_n(u_j)$
  has width at least

  \begin{equation*}
    \tfrac{k-j+1}{k} c(\mu n)^p +c(\varphi n)^p.
  \end{equation*}
\end{claim}
\begin{proof}
  We do another proof by induction, on $j$ (within the outer induction proof).
  Let $\pi(\Gamma)$ denote the root-to-leaf path in $T(u_j)$ that is vertically aligned in $\Gamma$. 
  Let $\pi_{j\rightarrow k}$ denote the path $\Seq{u_j,\ldots,u_k}$.
  Consider the last node $w$ that is common to both paths $\pi(\Gamma)$ and $\pi_{j\rightarrow k}$.

\medskip\noindent {\sc Case 1:}
$w=u_k$.
Let $\alpha$ and $\beta$ be the largest subtree among $\alpha_j,\ldots,\alpha_{k-1}$
and $\beta_j,\ldots,\beta_{k-1}$, respectively (in the special case $j=k$, let $\alpha=\beta=\emptyset$).
By Lemma~\ref{lem:subtrees}, 

  \begin{equation*}
    |\alpha|^p,|\beta|^p
    \geq \tfrac{k-j-1}{k} (\mu n)^p.
  \end{equation*}
  If $\pi(\Gamma)$ contains the left child of $u_k$, then the drawings of $\alpha$ and $R$ are separated by the vertical line through $\pi(\Gamma)$, and so (by the outer induction hypothesis) the overall drawing $\Gamma$ has width at least

  \begin{equation*}
    c|\alpha|^p+c|R|^p+1
    \:\geq\: \tfrac{k-j-1}{k} c(\mu n)^p +  c(\varphi n)^p+1
    \:\ge\: \tfrac{k-j+1}{k} c(\mu n)^p +  c(\varphi n)^p,
  \end{equation*}
   for a sufficiently small $c$ (since $\frac 1k (\mu n)^p = O(1)$).
  If $\pi(\Gamma)$ contains the right child of $u_k$, then $\beta_i$ and $L$ are vertically separated,  and the argument is similar.
    
\medskip\noindent {\sc Case 2:} $w=u_m$ for some $j\le m<k$.
Let $\alpha$ be the largest subtree among $\alpha_j,\ldots,\alpha_{m-1}$  (in the special case $m=j$, let $\alpha=\emptyset$).
By Lemma~\ref{lem:subtrees}, 

  \begin{equation*}
    |\alpha|^p
    \geq \tfrac{m-j-1}{k} (\mu n)^p.
  \end{equation*}
  Since $\pi(\Gamma)$ contains the left child of $u_m$, we know that the drawings of $\alpha$ and $T_n(u_{m+1})$ are separated by the vertical line through $\pi(\Gamma)$, and so by the induction hypotheses, the overall drawing $\Gamma$ has width at least
 
\PAPER{
  \begin{equation*}
    c|\alpha|^p+ \tfrac{k-m}{k}c(\mu n)^p + c(\varphi n)^p+1
    \:\geq\: \tfrac{k-j-1}{k} c(\mu n)^p +  c(\varphi n)^p+1
    \:\ge\: \tfrac{k-j+1}{k} c(\mu n)^p +  c(\varphi n)^p,
  \end{equation*}
}\LNCS{
  \begin{eqnarray*}
    c|\alpha|^p+ \tfrac{k-m}{k}c(\mu n)^p + c(\varphi n)^p+1
    &\geq& \tfrac{k-j-1}{k} c(\mu n)^p +  c(\varphi n)^p+1\\
    &\ge& \tfrac{k-j+1}{k} c(\mu n)^p +  c(\varphi n)^p,
  \end{eqnarray*}
}
  for a sufficiently small $c$.

\medskip\noindent {\sc Case 3:} $w=v_m$ for some $m<k$.  This is similar to Case 2.
\qed
\end{proof}

Applying Claim~\ref{claim:wst-indct} with $j=1$,
we see that any LR drawing of $T_n$ has width at
least $c(\mu n)^p + c(\varphi n)^p$, which is at least
$cn^p$, completing the induction proof, provided that

\begin{equation}
  \label{eqn:power-constraint}
  \varphi^p + \mu^p \geq 1.
\end{equation}

For $p=0.429$, we can choose $\mu=0.122$ and
$\varphi\approx 0.297513$ and verify (using a calculator) that both
\eqref{eqn:linear-constraint} and 
\eqref{eqn:power-constraint} are satisfied.

\begin{theorem}
  \label{thm:lower-bound}
  For every positive integer $n$, there is a binary tree with $n$ nodes such that any LR drawing requires $\Omega(n^{0.429})$ width.
\end{theorem}

\begin{remarks}
The maximum value of $p$ that guarantees the existence of $\mu$ and $\varphi$ satisfying 
\eqref{eqn:linear-constraint} and 
\eqref{eqn:power-constraint} has a concise description: it is given by $p=1/(1+x)$, where $x$ is the solution to the equation

\begin{equation*}
  1-2^{-x}=(2^{1/x}-1)^x.
\end{equation*}

Our lower-bound proof is very similar to Frati, Patrignani, and Roselli's~\cite{FratiPR20}, but there are two main differences: First, their tree construction was parameterized by a different parameter $h$ instead of $n$;
they upper-bounded the size $n$ by an exponential function on $h$ and lower-bounded the width by another exponential function on $h$.  Second, and more crucially, they chose $\varphi=\mu$ (in our terminology).  Besides convenience, we suspect that their choice was due to the above parameterization issue.  With this extra, unnecessary constraint $\varphi=\mu$, the best choice of $p$ was only around $0.418$.
\end{remarks}

\section{Final Remarks}

The main open problem is to narrow the remaining small gap in the exponents of the upper and lower bound (between $0.437$ and $0.429$).
The fact that both the upper and lower bound proofs use similar ``ruler patterns'' suggests that we are on the right track (even though looking for further tiny improvements in the upper-bound proof by complicating the analysis, along the lines of Section~\ref{sec:slight}, doesn't seem very worthwhile).

Frati {\it et al.}~\cite{FratiPR20} have computed the exact optimal width for small values of $n$, and according to their experimental data for all $n\le 455$, a function of the form $W^*(n)=an^b-c$ with the least-squares fit is $W^*(n)\approx 1.54n^{0.443}-0.55$.  Our results reveal that the true exponent is actually smaller.


Another open problem is to  bound the related function $W^{**}(n)$ mentioned in footnote~1 of the introduction; our new upper-bound proof does not work for this problem, but Chan's $O(n^{0.48})$ upper bound~\cite{SODA99} still holds.



\PAPER{
\small
\bibliographystyle{plain}
}\LNCS{
\bibliographystyle{splncs04}
}
\bibliography{tree}

\end{document}